\documentclass{article}
\usepackage[preprint]{icml2026}

\usepackage{graphicx}
\usepackage{float}
\usepackage{booktabs}
\usepackage{pgfplots}
\usepackage{amsmath,amssymb,amsthm}
\usepackage{stmaryrd}
\usepackage{multirow}
\usepackage{url}
\usetikzlibrary{arrows.meta,positioning,calc,decorations.pathreplacing,shapes.arrows}
\pgfplotsset{compat=1.18}
\definecolor{scienceblue}{RGB}{76,114,176}
\definecolor{sciencered}{RGB}{196,78,82}
\definecolor{sciencepurple}{RGB}{130,80,180}

\icmltitlerunning{HyperParallel-Mpipe: A Composable Algebra System for Optimizing MLLM Training over Supernode Clusters}

\newtheorem{lemma}{Lemma}
\newtheorem{corollary}{Corollary}
\allowdisplaybreaks

\begin{document}
\twocolumn[
\icmltitle{HyperParallel-Mpipe: A Composable Algebra System for \\Optimizing MLLM Training over Supernode Clusters}

\begin{center}
\large
Chong Li$^{1}$, Zhengdao Yu$^{1}$, Nelson Lossing$^{1}$, Thibaut Tachon$^{1}$,\\ Pierre Leca$^{1}$, Etienne Filhol$^{1}$, Yujie Yuan$^{1}$, Chong Bao$^{2}$, Teng Su$^{2}$

\vskip 0.1in
$^{1}$Huawei Fourier Research Center, Paris, France\\
$^{2}$Huawei Technologies Co., Ltd, Hangzhou, China
\end{center}
\vskip 0.2in

\begin{abstract}
Modern AI applications have expanded beyond text-only interaction into a wide
range of multimodal scenarios, making multimodal large language models (MLLMs)
crucial for both research and industry. However, compared with traditional
decoder-only LLM training, large-scale MLLM training often shows much lower MFU.
We analyze the key pain points in MLLM training and introduce Mpipe, which uses
a schedule algebra to derive concrete runtime behavior from a compact schedule 
specification. From this algebra, Mpipe derives transpose, a
multimodal-aware heterogeneous parallel schedule that remaps modality-encoder
computation into otherwise idle pipeline regions.
On Ascend 910C NPU clusters, Mpipe achieves 2.70$\times$ speedup in a 
small-scale setting and 1.21$\times$ speedup in a 512-card large-scale setting.
\end{abstract}

]

\makeatletter
\icml@noticeprintedtrue
\makeatother

\section{Introduction}
Modern AI applications are no longer limited to text-only interaction. In many widely used scenarios, such as visual question answering, document understanding, video analysis, embodied agents, and content generation, 
models are expected to jointly understand language together with images, videos, and other modal inputs. This trend has made multimodal large language models (MLLMs) an increasingly important direction in both research and industry. 

As model parameters rapidly scale, training frontier models increasingly depends
on large accelerator clusters. For example, Llama~3 reports production training
on clusters with up to 16K GPUs~\cite{grattafiori2024llama3}. MLLMs further
increase this demand beyond decoder-only LLMs. In addition to the language
backbone, they must process high-resolution images, videos, and other modality
inputs through extra model components, which increases computation, memory
footprint, and data movement.

However, large-scale MLLM training often achieves much lower MFU (model FLOPs 
utilization) than decoder-only LLM training~\cite{zhang2025disttrain}. 
Figure~\ref{fig:mfu-gap} compares the MFU of representative LLM and MLLM 
training workloads. 

The dominant source of this gap is the modality encoder, which sits at the input
side of the pipeline: its highly variable work depends only on data available
before the LLM pipeline begins. Such source-side variance need not be chased at
runtime, by per-iteration schedule search~\cite{xue2026dip} or data-dependent
parallelism~\cite{megascaleomni}. It can instead be relocated off the critical
path into the pipeline's intrinsic warmup bubbles by a single static placement.
We propose Mpipe, which formalizes parallel schedules as a small composable
algebra and derives \emph{transpose}, a schedule that places encoder computation
in those bubbles, as one point in that algebra. The resulting schedule stays
static and invariant to the modality mix, with no runtime scheduling overhead.
Our contributions are:
\begin{itemize}
\setlength{\itemsep}{1pt}
\setlength{\topsep}{2pt}
\setlength{\parsep}{0pt}
\setlength{\partopsep}{0pt}
\item \textbf{A schedule algebra for parallel training.} We formalize parallel
pipeline schedules as a small composable algebra in which a schedule is a list
of per-region skeletons, and a single derivation maps a cut, schedule, and model
to concrete runtime behavior: device placement, collective communication,
and execution order, from which we read a cost model that predicts the step makespan. 
We prove a backward-footprint lemma and a schedule-invariance corollary: one schedule 
is correct whether the modality encoder is frozen or trained.
\item \textbf{Transpose.} From this algebra Mpipe derives \emph{transpose}, a
static schedule that replicates the encoder across pipeline ranks, runs it inside
the warmup bubbles, and gathers its output into the first LLM stage. The
placement is decided once, adds no runtime scheduling overhead, and is invariant
to the modality mix.
\item \textbf{End-to-end results.} On Ascend 910C NPU clusters, Mpipe delivers
2.70$\times$ speedup in a small-scale setting and 1.21$\times$ at 512 cards, with
no change to the training loss.
\end{itemize}

\begin{figure}[t]
\centering
\begin{tikzpicture}
\begin{axis}[
    width=\columnwidth,
    height=4.9cm,
    ybar,
    bar width=14pt,
    ymin=0,
    ymax=65,
    ytick={0,20,40,60},
    yticklabels={0\%,20\%,40\%,60\%},
    symbolic x coords={llama,deepseek,qwenvl,vitllama},
    xtick={llama,deepseek,qwenvl,vitllama},
    xticklabels={
        {LLM\\LLaMA3},
        {LLM\\DeepSeek},
        {MLLM\\QwenVL2},
        {MLLM\\ViT+LLaMA3\\+SD}
    },
    x tick label style={font=\scriptsize, align=center, text width=1.45cm},
    title={MFU (Model FLOPs Utilization)},
    title style={font=\normalsize},
    axis lines*=left,
    axis line style={draw=none},
    tick style={draw=none},
    ymajorgrids=true,
    major grid style={draw=gray!25},
    enlarge x limits=0.12,
    nodes near coords={\pgfmathprintnumber[fixed,precision=2,zerofill=false]{\pgfplotspointmeta}\%},
    every node near coord/.append style={font=\small, yshift=2pt, black},
    point meta=y,
    legend style={draw=none}
]
\addplot+[draw=none, fill=scienceblue, bar shift=0pt] coordinates {(llama,41)};
\addplot+[draw=none, fill=scienceblue, bar shift=0pt] coordinates {(deepseek,44.57)};
\addplot+[draw=none, fill=sciencered, bar shift=0pt] coordinates {(qwenvl,13.5)};
\addplot+[draw=none, fill=sciencered, bar shift=0pt] coordinates {(vitllama,20)};
\end{axis}
\end{tikzpicture}
\caption{MFU comparison between representative LLM and MLLM training workloads. The MLLM measurements are sourced from OrchMLLM~\cite{zheng2025orchmllm} and DistTrain~\cite{zhang2025disttrain}.}
\label{fig:mfu-gap}
\end{figure}
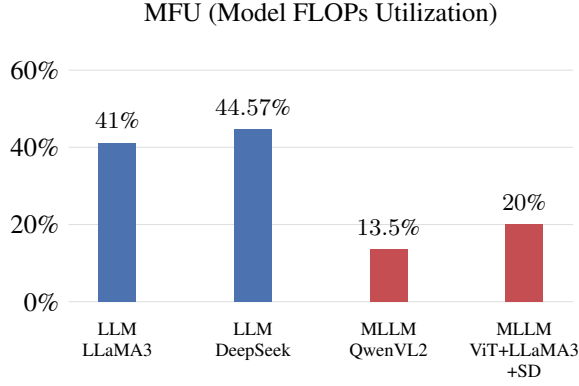

\section{From LLM to MLLM Training}

Beyond traditional NLP tasks, MLLMs must handle diverse inputs such as images,
videos, or audio before the language model can reason over them. A typical MLLM
therefore contains modality encoders for parsing these inputs, an LLM backbone
for reasoning, and sometimes modality generators for output synthesis. Between
these components, projector layers are commonly used to align modality
representations with the LLM backbone feature space (e.g.,
NExT-GPT~\cite{wu2023nextgpt}). MLLM training also often proceeds in multiple
stages: to align feature spaces across modalities, some modules may be frozen
while projector or alignment layers are updated.

The key structural change is that MLLMs no longer execute a single homogeneous
Transformer stack. They separate modality parsing from language reasoning, and
this separation naturally introduces components with different functional roles.
A ViT-style encoder~\cite{dosovitskiy2021image} extracts and aligns modality
features, the aligned tokens are then reasoned over in the LLM backbone, and the
reasoning result may be passed to modality generators such as
U-Net~\cite{ronneberger2015unet} or DiT~\cite{peebles2023scalable} to synthesize
modality-specific outputs. Therefore, the encoder, projector, LLM backbone, and
generator have naturally different responsibilities and process different
workloads; they are not interchangeable copies of the same repeated block, but
form a heterogeneous model pipeline.

The input side adds another source of difficulty. Modality content is often much
heavier than text: a long text sequence may only occupy a few kilobytes, while a
high-resolution image can occupy tens of megabytes and high-definition video can
reach gigabytes. It is also more variable: images may have different
resolutions, videos may have different lengths, and some samples may be pure
text with no additional modality content at all. After preprocessing, these
differences become different modality token sequence lengths. Since
attention-based encoder computation depends strongly on sequence length,
high-resolution images or long videos can require much more computation and
activation memory than small images or text-only samples. In MLLM training, both
the ratio and the workload of modality content can vary across samples. For
example, vision-language corpora often mix images with very different
resolutions, such as high-resolution photographs and low-resolution web images,
making the encoder workload vary dynamically from one microbatch to the next.

\begin{table}[t]
\centering
\caption{Approximate component heterogeneity under one MLLM training
configuration. ViT computation is measured with visual length fixed to $4\times$
the LLM backbone sequence length, where the LLM backbone sequence length is 8K.}
\label{tab:component-heterogeneity}
\resizebox{\columnwidth}{!}{%
\begin{tabular}{l c c c}
\toprule
Component & Parameter memory & Activation memory & Computation \\
 & (per layer) & (recomputed) & (FLOPs) \\
\midrule
LLM backbone & $\sim$4.95 GB & $\sim$1.79 GB & $\sim 4.76\times10^{18}$ \\
ViT encoder & $\sim$0.01 GB & $\sim$1.63 GB & $\sim 2.88\times10^{17}$ \\
\bottomrule
\end{tabular}%
}
\end{table}

At the component level, modality encoders are often much narrower than the LLM
backbone, but they can still introduce comparable activation memory and
non-negligible computation. Table~\ref{tab:component-heterogeneity} gives one
representative profile: the LLM backbone dominates parameter memory, while the
ViT encoder has much smaller parameter memory but comparable activation memory
and non-negligible computation. In our 512-NPU training traces, under pipeline
parallel (PP) execution, microbatches with few or no images spend only about 
200--300 ms in encoder-stage computation, whereas image-heavy microbatches can 
take more than 1 s. Such instability inserts dynamic bubbles into the pipeline.

In summary, efficient MLLM training faces two challenges. First,
\emph{component heterogeneity}: modality encoders and LLM backbones have
different memory and computation profiles. Second, \emph{dynamic encoder
workloads}: multimodal inputs change the amount of encoder computation and
activation memory at microbatch granularity. These challenges motivate the
system-level comparison in the next section.

\section{Limitations of Existing MLLM Training Systems}\label{sec:limits}

Training MLLMs in practice usually relies on AI training frameworks, such as
PyTorch~\cite{li2020pytorchddp}, Megatron-LM~\cite{shoeybi2019megatron}, and
Hyper-Parallel\footnote{\url{https://atomgit.com/mindspore/hyper-parallel}}.
These frameworks manage low-level training details and expose reusable
parallelism templates for deploying large models. This abstraction works well
for repeated Transformer stacks, but becomes challenging when applied to
heterogeneous and dynamic MLLM workloads.

Tensor parallelism~\cite{shoeybi2019megatron} is effective for splitting large
hidden matrices in a repeated Transformer stack, but the same partition may
provide little memory relief for a much narrower encoder. Pipeline
parallelism~\cite{huang2019gpipe} distributes layers across stages, but encoder
stages and LLM backbone stages can carry very different workloads, injecting
substantial bubbles into the pipeline. Dynamic encoder workloads further
complicate standard parallel execution and amplify the pipeline effect above.
Data parallelism~\cite{dean2012large} assigns different input subsets to
different replicas, so an image-heavy replica can become a straggler while a
text-only replica waits. Sequence and context
parallelism~\cite{korthikanti2022reducing,liu2023ringattention} can reduce
long-sequence attention memory, but applying the same sequence split to small or
empty modality inputs may introduce communication without much useful memory
reduction.

The preceding section identifies two challenges for efficient MLLM training:
component heterogeneity and dynamic encoder workloads. Beyond these
general-purpose templates, existing MLLM training systems approach these
challenges from three directions.

\textbf{Spatial resource partitioning.}
DistTrain~\cite{zhang2025disttrain} assigns hardware resources separately to
encoders, LLM backbones, and generators according to their module-level load, and
customizes parallel strategies for different modules. This cluster-level
partitioning can match a specific encoder workload, but such load-specific
partitioning is hard to keep effective under dynamic encoder workloads.

\textbf{Temporal bubble exploitation.}
Optimus~\cite{feng2025optimus} adjusts the temporal schedule to exploit LLM
backbone pipeline bubbles for MLLM training. It first analyzes computation
dependencies, then profiles encoder computation, decomposes it into kernel-level
units, and temporally schedules these units into LLM backbone pipeline bubbles. This fine-grained scheduling can reduce exposed
bubbles by overlapping encoder computation with otherwise idle pipeline slots.
However, this schedule depends on the profiled encoder workload. When encoder
workloads change dynamically, the values obtained from a static profile become
less reliable, making the pre-planned overlap less effective and leading to
suboptimal results.

\textbf{Runtime load balancing.}
DIP~\cite{xue2026dip} partitions the encoder and backbone into dedicated,
modality-aware pipeline segments across the same ranks to alleviate the load
imbalance between encoder stages and LLM backbone stages. It also prefetches
metadata for upcoming microbatches and splits large modality workloads to
improve balance. Nevertheless,
DIP still keeps the encoder coupled inside the LLM backbone pipeline. It mitigates the
mismatch through this placement and runtime load balancing, but the
architectural heterogeneity of MLLMs remains part of the pipeline itself.
Moreover, runtime metadata prefetching, workload splitting, and schedule search
are overlapped off the critical path at the scales DIP evaluates, so their cost
at much larger cluster scale remains undemonstrated.

Across these three directions, the encoder remains coupled inside the LLM
backbone pipeline, and each system pays for heterogeneity and dynamic workloads
through profiling, search, or runtime balancing. As a result, a static plan must
either be recomputed every iteration or risks being invalidated by the next
microbatch. Mpipe takes a different direction: it relocates encoder computation
out of the critical path and into the pipeline's intrinsic warmup bubbles through
a single static placement, so it requires no per-iteration profiling or search
and absorbs the dynamic workload rather than planning around it. We present this
design in the next section.

\section{Mpipe}\label{sec:overview}

Mpipe trains MLLMs, which pair a modality encoder with an LLM backbone. The two
blocks are heterogeneous, and that is the difficulty. The encoder is a source
block. Its cost depends on the input, such as image resolution, sequence length,
or the number of tiles. The LLM layers are uniform. If the encoder runs as a
normal pipeline stage, its variance is exposed on the critical path: it stretches
the stage it occupies, and the resulting bubble spreads down the pipeline.

Mpipe removes this exposure with \emph{transpose}. Transpose replicates the
encoder on every pipeline rank and runs each rank's encodes inside the warmup
bubble. Encoder work then overlaps the pipeline fill instead of taking its own
stage (Figure~\ref{fig:transpose}). We derive transpose from a small
\emph{schedule algebra} (\S\ref{sub:algebra}). A schedule is a list of per-region
\emph{skeletons}, one skeleton per depth region. A derivation maps it to runtime
behavior and fixes three things: where each region runs (placement), which
collectives cross region boundaries, and the order of forward and backward
events, including the warmup and cooldown slots. In this algebra Mpipe is
$\langle\mathsf{transpose},\mathsf{1f1b}\rangle$: a transpose skeleton on the
encoder, 1F1B on the backbone. The derivation gates backward work on the model's
trainable bits, not on the schedule. So one transpose schedule serves a frozen or
a trained encoder, unchanged. Section~\ref{sub:transpose} develops transpose as a
method.

\begin{figure}[t]
\centering
\includegraphics[width=\columnwidth]{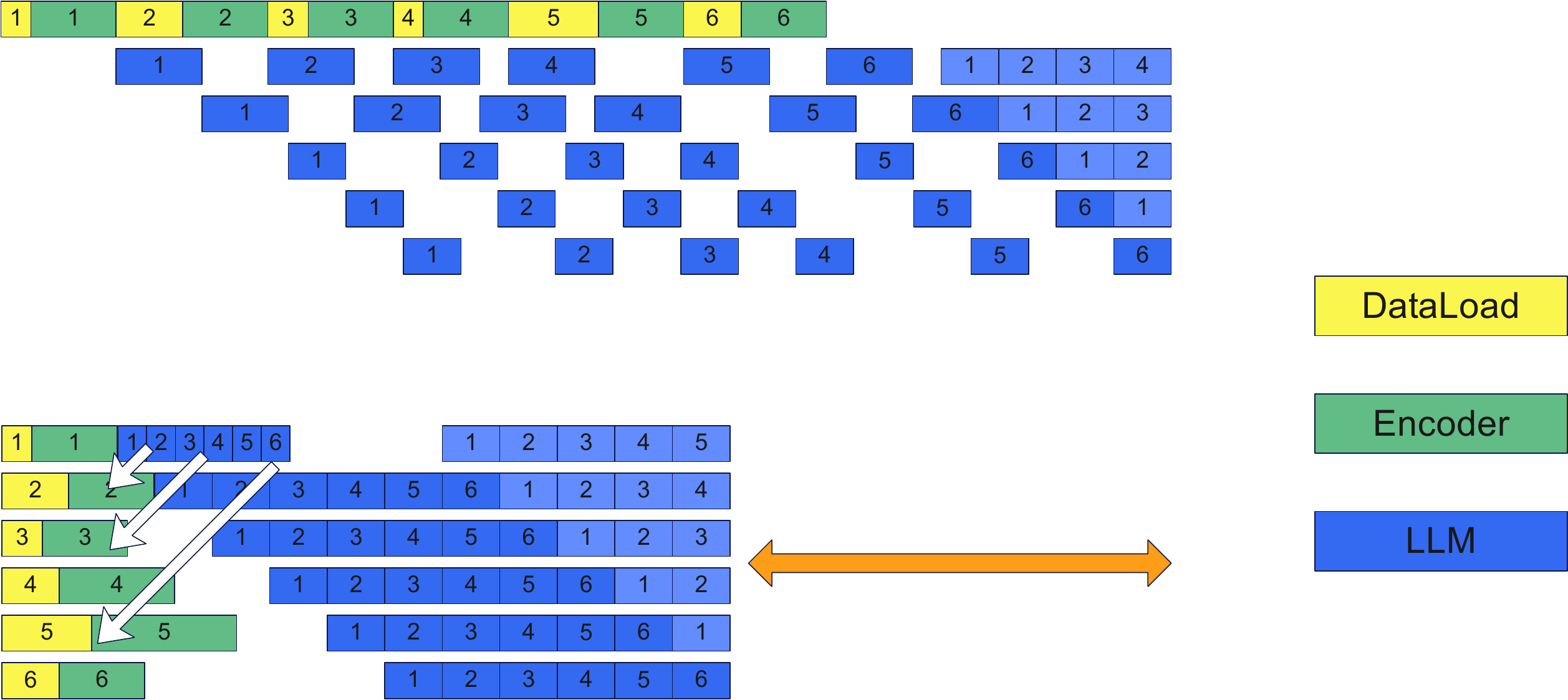}
\caption{Transpose overview. Instead of exposing modality encoder computation as
a regular pipeline stage, Mpipe transposes encoder execution into otherwise idle
warmup bubbles of the LLM pipeline.}
\label{fig:transpose}
\end{figure}

\subsection{Schedule Algebra}\label{sub:algebra}

A schedule applies one skeleton to each depth region of the model. A
\textit{model} is a sequence of layers $U=\langle u_1,\dots,u_L\rangle$. The
algebra does not depend on a region's role. This paper uses the encoder /
LLM-backbone split and leaves generator (sink) regions to future work. Two inputs
define a parallelization. A \textit{cut} sets $k{-}1$ interior boundaries
$c_1<\dots<c_{k-1}$ in $(0,L)$, with $c_0=0$ and $c_k=L$. It splits $U$ into $k$
contiguous regions $\Pi_1,\dots,\Pi_k$, one per interval. A \textit{schedule}
$\sigma=\langle\sigma_1,\dots,\sigma_k\rangle$ gives one skeleton to each region.
Composition is region concatenation. Each seam in Table~\ref{tab:seam} is a layout
change. The pair $\mathsf{Replicated}\!\to\!\mathsf{Replicated}$ is the only one
with no change: two replicated regions share one $\mathsf{ByMicrobatch}$ owner, so
the seam is the identity and the two regions collapse into one. We therefore drop
this seam and forbid adjacent replicated regions. Every valid schedule then has a
normal form in which transposes are separated by a sharded region. A schedule is
\textit{valid} when its cut tiles $U$ and it is in normal form. Both checks use
$(c,\sigma)$ only. So $\langle\mathsf{transpose},\mathsf{1f1b}\rangle$ and
$\langle\mathsf{transpose},\mathsf{1f1b},\mathsf{transpose}\rangle$ are valid,
while $\langle\mathsf{transpose},\mathsf{transpose}\rangle$ reduces to
$\langle\mathsf{transpose}\rangle$. The function $\operatorname{derive}$ maps a
valid $(c,\sigma)$ and the model to runtime behavior.

\begin{align*}
\textit{skeleton} \;&::=\; \mathsf{1f1b} \mid \mathsf{gpipe} \mid \mathsf{transpose} \\
\textit{schedule}  \;&::=\; \langle \textit{skeleton}_1,\dots,\textit{skeleton}_k \rangle \\
\operatorname{derive} \;&:\; \textit{cut} \times \textit{schedule} \times \textit{model} \\
   &\phantom{:}\;\;\to\; \textit{placement} \times \textit{collectives} \times \textit{order}
\end{align*}

A skeleton fixes a region's forward/backward order. Its one derived property is
the \textit{layout}: $\operatorname{layout}(\mathsf{transpose})=\mathsf{Replicated}$,
and $\operatorname{layout}(\mathsf{1f1b})=\operatorname{layout}(\mathsf{gpipe})=
\mathsf{Sharded}$. We write $\operatorname{layout}_j:=\operatorname{layout}(\sigma_j)$
for region $j$. A $\mathsf{Sharded}$ region pipelines its layers across the ranks.
A $\mathsf{Replicated}$ region ($\mathsf{transpose}$) puts the whole region on
every rank and runs data parallelism over the microbatches. Its forwards run in
the warmup bubble, and its backwards, when trainable, in the cooldown. Different
schedules are points in one space. VPP-3 is $\langle\mathsf{1f1b},\mathsf{1f1b},
\mathsf{1f1b}\rangle$: three $\mathsf{Sharded}$ regions woven into one wavefront.
Mpipe is $\langle\mathsf{transpose},\mathsf{1f1b}\rangle$.

\emph{Notation.} $P$ is the number of pipeline ranks
($\textit{rank}\in\{0,\dots,P{-}1\}$) and $M$ the number of microbatches
($\textit{mb}\in\{1,\dots,M\}$). A $\textit{region}$ is in $\{1,\dots,k\}$, and a
$\textit{slab}$ is a contiguous layer interval. $[X]$ is a list of $X$. $X^{?}$ is
an optional $X$, present only when trainability allows it (Lemma~\ref{lem:gate}).
The derivation reads the model in two places: each region's layer count, which
with the cut gives the per-rank slabs, and its $\operatorname{trainable}$ bit,
which gates all backward work. It then emits three outputs.

\begin{align*}
\textit{layout} \;&::=\; \mathsf{Sharded} \mid \mathsf{Replicated} \\
\textit{owner}  \;&::=\; \mathsf{ByLayer}(\textit{rank}\!\to\!\textit{slab}) \\*
                & \quad\;\;\, \mid \mathsf{ByMicrobatch}(\textit{mb}\!\to\!\textit{rank}) \\
\textit{comm}   \;&::=\; \mathsf{Gather} \mid \mathsf{Scatter} \mid \mathsf{Send} \mid \mathsf{RecvGrad} \\
\textit{group}  \;&::=\; \mathsf{ReplicaGroup} \mid \mathsf{DataGroup} \\
\textit{seam}   \;&:\; \textit{region} \times \textit{region} \\
\textit{event}  \;&::=\; \mathsf{Fwd}(\textit{region},\textit{rank},\textit{mb}) \\*
                & \quad\;\, \mid \mathsf{Bwd}(\textit{region},\textit{rank},\textit{mb}) \\*
                & \quad\;\, \mid \mathsf{Coll}(\textit{comm},\textit{seam},\textit{mb}) \\*
                & \quad\;\, \mid \mathsf{DpReduce}(\textit{region},\textit{group}) \\*
                & \quad\;\, \mid \mathsf{Loss}(\textit{mb}) \mid \mathsf{StepBarrier} \\
\textit{edge}   \;&::=\; \mathsf{Activation} \mid \mathsf{Gradient} \mid \mathsf{Turnaround} \\*
                & \quad\;\;\, \mid \mathsf{Accumulate} \mid \mathsf{Sequence} \\[4pt]
\textit{placement}   \;&:\; \textit{region} \to \textit{layout} \times \textit{owner} \\
\textit{collectives} \;&:\; \{\, \textit{seams}{:}\,[\textit{seam}\!\times\!\textit{comm}\!\times\!\textit{comm}^{?}], \\*
                     & \qquad\;\; \textit{reduces}{:}\,[\textit{region}\!\times\!\textit{group}] \,\} \\
\textit{order}       \;&:\; \{\, \textit{nodes}{:}\,\textit{rank}\!\to\![\textit{event}], \;\; \textit{edges}{:}\,[\textit{edge}] \,\}
\end{align*}

The unit of placed work is a \emph{stage}: the pair $(\textit{region},\textit{rank})$
that one rank runs. The runtime numbers it
$\textit{stage\_index}=\textit{rank}+(\textit{region}{-}1)P$, so VPP-3's three
regions are the three virtual stages each rank interleaves. A $\mathsf{Sharded}$
region emits one $\mathsf{Fwd}(\textit{region},r,m)$ per rank $r$, on its
$\mathsf{ByLayer}$ slab. A $\mathsf{Replicated}$ region emits one at
$r=\mathsf{ByMicrobatch}(m)$. \emph{Placement} gives each region its layout and an
\emph{owner} that splits its work; the $\mathsf{ByMicrobatch}$ owner assigns a
replicated region's microbatches to ranks. \emph{Collectives} is a derived view of
adjacent layouts. Each seam's forward $\textit{comm}$ comes from the layout pair
(Table~\ref{tab:seam}). Its backward $\textit{comm}^{?}$ fires under
Lemma~\ref{lem:gate}. The $\textit{reduces}$ field records each trainable region's
weight reduction over its $\textit{group}$: $\mathsf{ReplicaGroup}$ for a
transposed region, $\mathsf{DataGroup}$ for a sharded one. A $\mathsf{DataGroup}$
reduce is a no-op when that axis has one member, i.e.\ pure pipeline parallelism.
\emph{Order} is one dependency graph over the stage-indexed events. $\textit{nodes}(r)$
is its projection onto rank $r$, and $\textit{edges}$ are its five typed arcs. The
graph is acyclic by construction, so there is no cyclic-wait deadlock, and it
fixes the forward/backward turnaround and every activation's live interval.

\begin{table}[htbp]
\centering
\caption{Seam collective for each pair of adjacent layouts. The two cross-layout
rows are mirrors. $\mathsf{Replicated}\!\to\!\mathsf{Replicated}$ has no layout
change, so its seam is the identity. We omit it and forbid adjacent replicated
regions. This is the only missing entry, and it is what makes validity, and hence
composition, non-trivial.}
\label{tab:seam}
\small
\begin{tabular}{l l l}
\toprule
$\operatorname{layout}_j\!\to\!\operatorname{layout}_{j+1}$ & fwd & bwd \\
\midrule
$\mathsf{Replicated}\!\to\!\mathsf{Sharded}$    & $\mathsf{Gather}$  & $\mathsf{Scatter}$ \\
$\mathsf{Sharded}\!\to\!\mathsf{Sharded}$       & $\mathsf{Send}$    & $\mathsf{RecvGrad}$ \\
$\mathsf{Sharded}\!\to\!\mathsf{Replicated}$    & $\mathsf{Scatter}$ & $\mathsf{Gather}$ \\
$\mathsf{Replicated}\!\to\!\mathsf{Replicated}$ & \multicolumn{2}{l}{identity (omitted)} \\
\bottomrule
\end{tabular}
\end{table}

\begin{lemma}[Backward footprint]\label{lem:gate}
A region's backward work is fixed by the model, through three trainability gates.
Order the regions input-to-loss as $\Pi_1,\dots,\Pi_k$. For $\Pi_j$:
\emph{(i)} $\Pi_j$ \emph{sends} its input-activation gradient upstream (the
reverse collective on its input seam $(\Pi_{j-1},\Pi_j)$) iff
$\exists\, i<j:\operatorname{trainable}(\Pi_i)$;
\emph{(ii)} $\Pi_j$ \emph{receives} a gradient from downstream (the reverse
collective on its output seam $(\Pi_j,\Pi_{j+1})$) iff
$\exists\, i\le j:\operatorname{trainable}(\Pi_i)$;
\emph{(iii)} $\Pi_j$ reduces weight gradients iff $\operatorname{trainable}(\Pi_j)$.
The send and receive gates differ by one region. $\Pi_j$ receives but does not
send exactly when it is trainable and no region before it is: the trainable
\emph{source}.
\end{lemma}

\noindent\emph{Proof.} A region's input-activation gradient is used only upstream,
by $\Pi_{j-1}$ and earlier. A frozen region still passes it through. So $\Pi_j$
sends iff some region before it is trainable ($\exists\,i<j$). $\Pi_j$ receives
iff it runs any backward at all, that is iff $\Pi_j$ or some region before it is
trainable ($\exists\,i\le j$). The two index sets differ only at $i=j$. Finally,
$\Pi_j$'s receive is $\Pi_{j+1}$'s send (the collective carries exactly
$\Pi_{j+1}$'s input-activation gradient), so the seam gate and the region gate are
one predicate. \hfill$\square$

\begin{corollary}[Schedule-invariance]\label{cor:inv}
Gates (i) through (iii) use only $\operatorname{trainable}(\cdot)$ and the cut,
never $\sigma$. So a region's backward footprint does not depend on the schedule:
the encoder has the same footprint whether it is transposed ($\mathsf{Replicated}$)
or pipelined ($\mathsf{Sharded}$). The footprint says \emph{which} backward events
fire. The collective \emph{type} for each one still depends on $\sigma$ through
Table~\ref{tab:seam}. One $\mathsf{transpose}$ schedule therefore needs no
per-trainability variant: $\operatorname{derive}$ reads $\operatorname{trainable}$,
not $\sigma$.
\end{corollary}

\paragraph{Instantiation.}
Take Mpipe, $\langle\mathsf{transpose},\mathsf{1f1b}\rangle$. The encoder is
$\mathsf{enc}=\Pi_1$ ($\mathsf{Replicated}$) and the backbone is $\mathsf{llm}=\Pi_2$
($\mathsf{Sharded}$). A $(\cdot)^{?}$ is present only when the encoder is trainable.

\begin{align*}
&\textit{placement}: \;\;\mathsf{enc} \mapsto (\mathsf{Replicated},\mathsf{ByMicrobatch}),\\*
&\hphantom{\textit{placement}:} \;\;\;\mathsf{llm} \mapsto (\mathsf{Sharded},\mathsf{ByLayer}) \\[6pt]
&\textit{collectives}: \\*
&\quad \textit{seams} : [((\mathsf{enc},\mathsf{llm}),\, \mathsf{Gather},\, \mathsf{Scatter}^{?})] \\*
&\quad \textit{reduces} : [(\mathsf{enc},\mathsf{ReplicaGroup})^{?},\, (\mathsf{llm},\mathsf{DataGroup})] \\[6pt]
&\textit{order edges}\ (\text{microbatch } m,\ \textit{rank}\text{ dropped}): \\*
&\quad \mathsf{Fwd}(\mathsf{enc},m) \to \mathsf{Coll}(\mathsf{Gather},(\mathsf{enc},\mathsf{llm}),m) \\*
&\quad\quad \to \mathsf{Fwd}(\mathsf{llm},m) \to \mathsf{Loss}(m) \to \mathsf{Bwd}(\mathsf{llm},m), \\*
&\quad \mathsf{Bwd}(\mathsf{llm},m)\ \text{forks:} \\*
&\quad\quad \rightsquigarrow \mathsf{DpReduce}(\mathsf{llm},\mathsf{DataGroup}) \to \mathsf{StepBarrier} \\*
&\quad\quad \rightsquigarrow \mathsf{Coll}(\mathsf{Scatter},(\mathsf{enc},\mathsf{llm}),m)^{?} \to \mathsf{Bwd}(\mathsf{enc},m)^{?} \\*
&\quad\quad\quad \to \mathsf{DpReduce}(\mathsf{enc},\mathsf{ReplicaGroup})^{?} \to \mathsf{StepBarrier}
\end{align*}

The chain gives the $\textit{edges}$ for one microbatch, with the rank index dropped; 
$\rightsquigarrow$ marks a fork into independent branches. $\mathsf{Fwd}(\mathsf{enc},m)$ 
stands for $\mathsf{Fwd}(\mathsf{enc},\mathsf{ByMicrobatch}(m),m)$, and 
$\mathsf{Fwd}(\mathsf{llm},m)$ for the family over the $\mathsf{llm}$ slabs on every rank. 
The per-rank schedule $\textit{nodes}(r)$ is this graph projected onto rank $r$; 
Figure~\ref{fig:transpose} shows it schematically. Forward arcs are $\mathsf{Activation}$ 
edges. The $\mathsf{Loss}\to\mathsf{Bwd}$ pivot is a $\mathsf{Turnaround}$. Each 
$\mathsf{Bwd}$ emits an $\mathsf{Accumulate}$ edge to its $\mathsf{DpReduce}$. When its 
region has a trainable strict predecessor (Lemma~\ref{lem:gate}(i)), it also emits a
$\mathsf{Gradient}$ edge into its reverse seam. The two
$\mathsf{DpReduce}\to\mathsf{StepBarrier}$ arcs are independent $\mathsf{Sequence}$
edges, with no order between the two reduces. $\mathsf{Bwd}(\mathsf{llm},m)$ is
never gated, since the backbone is always trainable. But its $\mathsf{Gradient}$
edge, the $\mathsf{Scatter}^{?}$ that ships the LLM's input-activation gradient
back to the encoder, fires only when the encoder is trainable
(Lemma~\ref{lem:gate}(ii) at $\mathsf{enc}$ equals Lemma~\ref{lem:gate}(i) at
$\mathsf{llm}$). The encoder is the trainable \emph{source} of Lemma~\ref{lem:gate}.
It receives the $\mathsf{Scatter}^{?}$ and runs its weight backward
$\mathsf{Bwd}(\mathsf{enc},m)^{?}$. This is an $\mathsf{Accumulate}$ edge only: the
encoder has no predecessor, so no $\mathsf{Gradient}$ edge, and it sends nothing
upstream. The $\mathsf{Activation}$ edge from each warmup $\mathsf{Fwd}(\mathsf{enc},m)$
to its cooldown consumer is the long-lived encoder state.

\paragraph{A second derivation.}
$\operatorname{derive}$ is not tied to one schedule. VPP-3,
$\langle\mathsf{1f1b},\mathsf{1f1b},\mathsf{1f1b}\rangle$, places every region as
$(\mathsf{Sharded},\mathsf{ByLayer})$ and gives two interior seams,
$((\Pi_1,\Pi_2),\mathsf{Send},\mathsf{RecvGrad}^{?})$ and
$((\Pi_2,\Pi_3),\mathsf{Send},\mathsf{RecvGrad}^{?})$. The $\textit{seam}$ argument
of $\mathsf{Coll}$ is what tells them apart, which is why it is not optional. Each
rank owns one slab of all three regions, so the stage-indexed events are the
interleaved 1F1B wavefront. $\Pi_1$ again sends no input-activation gradient. The
same six events and five edges describe both schedules; only the cut and the
skeletons differ.

\paragraph{The same schedule, two footprints.}
Schedule-invariance (Corollary~\ref{cor:inv}) is one axis. Freezing is the other.
Freezing the encoder leaves $\sigma$, the placement, and the forward order
unchanged, but deletes every $(\cdot)^{?}$. The order collapses to
$\mathsf{Fwd}(\mathsf{enc},m)\to\mathsf{Coll}(\mathsf{Gather},\dots)\to
\mathsf{Fwd}(\mathsf{llm},m)\to\mathsf{Loss}(m)\to\mathsf{Bwd}(\mathsf{llm},m)\to
\mathsf{DpReduce}(\mathsf{llm},\mathsf{DataGroup})\to\mathsf{StepBarrier}$, with no
reverse $\mathsf{Scatter}$, no $\mathsf{Bwd}(\mathsf{enc})$, and no replica reduce.
The schedule comes from the user; the footprint comes from the model.

\paragraph{Memory as a derived view.}
Memory is read off $(\textit{placement},\textit{order})$ with no extra output. An
activation is live from when it is produced to its last use, and a rank's peak is
the running sum of live sizes. The transposed encoder is one long
$\mathsf{Activation}$ edge, live across the whole LLM span; transpose hides it in
the warmup bubble instead of paying for an extra sharded stage with the same
lifetime.

\paragraph{Applicability.}
The two composition axes, a depth cut with one skeleton per region and an owner
map that weaves $\mathsf{Sharded}$ regions into one wavefront, are enough for every
schedule we use: 1F1B, GPipe, interleaved VPP, and
$\langle\mathsf{transpose},\mathsf{1f1b}\rangle$. The algebra also extends in three
graded steps. A \emph{fold} (a reflected $\mathsf{ByLayer}$ owner plus a
$\mathsf{wave}$ skeleton) reaches Hanayo~\cite{hanayo}. An \emph{event split} (a
split backward plus a fused forward-backward node) reaches DualPipeV~\cite{dualpipev}.
Both are conservative: the present algebra is the degenerate case. A
\emph{bidirectional feed}, as in DualPipe~\cite{deepseekv3}, needs a new input and
stays outside. Our use-case sits inside the closed core. We treat $\mathsf{1f1b}$,
$\mathsf{gpipe}$, and $\mathsf{transpose}$ as composable parallel skeletons in the
sense of Cole~\cite{cole1989skeletons}, applied to the execution schedule of a
fixed model rather than to an algorithm.\footnote{The derivation is defined for
every valid cut. The present instantiation uses the two-region encoder/LLM cut;
intermediate cuts are future work.}

\subsection{Transpose}\label{sub:transpose}

Transpose is the $\mathsf{Replicated}$ skeleton of \S\ref{sub:algebra}, read as a
systems mechanism. Replicating the encoder on every rank is a scatter-gather step
in the sense of SGL~\cite{li2012sgl}. The replication scatters the microbatches
across the ranks. Each rank runs the encoder forward for the microbatches its
$\mathsf{ByMicrobatch}$ owner gives it, on its own, inside its warmup bubble. The
$\mathsf{Replicated}\to\mathsf{Sharded}$ seam then gathers the per-rank encoder
outputs into the first LLM stage (Figure~\ref{fig:transpose}). Encoder work
overlaps the pipeline fill instead of taking an extra stage on the critical path.
The owner map is a free choice, which Mpipe uses to match heavier encodes to
longer warmup slots.

This PP-internal replication is not the same as a separate data-parallel group
for the encoder. The encoder forward runs in the same pipeline group, and its
output feeds the first LLM stage at once through the seam gather. So Mpipe adds no
global aggregation and no cross-group synchronization on the encoder path. The
encoder output is used where it is made, inside the pipeline that needs it. We
keep the locality and independence of data-parallel work without a second
collective domain.

\subsection{Cost Model}\label{sub:cost}

The same two outputs that give memory, $\textit{placement}$ and $\textit{order}$,
also give a cost model. We weight each event with a duration
$\operatorname{cost}:\textit{event}\to\mathbb{R}_{\ge 0}$ and read the step time
off the order graph, so a schedule's cost is computed, not measured. A compute
event runs on one rank and costs the roofline of its slab; a collective moves an
$n$-byte payload over its group at the $\alpha$-$\beta$ rate of the SGL bridging
model~\cite{li2012sgl}:
\begin{align*}
\operatorname{cost}(\mathsf{Fwd})=\operatorname{cost}(\mathsf{Bwd}) &= \max\!\big(W/F,\ D/B\big), \\
\operatorname{cost}(\mathsf{Coll})=\operatorname{cost}(\mathsf{DpReduce}) &= \alpha + \beta\, n,
\end{align*}
where $W$ is the slab's FLOPs at peak rate $F$, $D$ its HBM traffic at bandwidth
$B$, and $n$ the payload bytes.

These weights turn $\textit{order}$ into a weighted graph in which each rank's
events follow their $\textit{nodes}(r)$ order. The step \emph{makespan} is its
longest path, and a rank's \emph{bubble} is its idle time:
\begin{align*}
T =             & \max_{p\,\in\,\mathrm{paths}(\textit{order})} \sum_{e\in p}\operatorname{cost}(e), \\
\mathrm{bub}(r) =&\quad T - \!\!\sum_{e\in\textit{nodes}(r)}\!\!\operatorname{cost}(e).
\end{align*}
SGL sums its supersteps in sequence; a pipeline overlaps them, so $T$ is a
longest path, not a sum of stage times. Because $T$ comes from the schedule
alone, it can be checked against the measured step time.

This makes $\mathsf{transpose}$ quantitative. Write $w_r$ for the warmup bubble on
rank $r$, its idle time before the first LLM forward, and $E_r$ for the encoder
work that rank runs,
$E_r=\sum_{m:\,\mathsf{ByMicrobatch}(m)=r}\operatorname{cost}(\mathsf{Fwd}(\mathsf{enc},r,m))$.
Transpose places $E_r$ in that window, so it hides $\min(E_r,w_r)$ and exposes
\begin{align*}
\mathrm{spill}_r = \max\!\big(0,\ E_r - w_r\big)
\end{align*}
on the critical path. When every $E_r\le w_r$ the encoder is free, and the step
pays only the worst rank's overflow, $\max_r \mathrm{spill}_r$. The
$\mathsf{ByMicrobatch}$ owner is the assignment that minimizes this exposed spill,
a min-max balancing of encoder work against per-rank warmup slack; matching
heavier encodes to longer warmup slots is its greedy form.

The model predicts where transpose helps. When the encoder is small against the
warmup slack, as when the backbone dominates at scale, $E_r\le w_r$ and the
encoder leaves the step untouched off the bubble; when it is large against the
slack, the residual $\max_r \mathrm{spill}_r$ is exposed and the gain shrinks.
This is the trend we measure in \S\ref{sec:exp}.

\paragraph{Compositional cost.}
The longest path above is a global property of the assembled graph, but it can
also be read compositionally. Give each schedule a \emph{timing} denotation
$\llbracket\sigma\rrbracket$: a function that maps the arrival times of its
inputs, one per rank and microbatch, to the completion times of its outputs.
Each such function is monotone and linear in the $(\max,+)$ semiring, since a
completion time is the maximum of an input-ready time and a resource-free time,
plus a duration. Concatenating two schedules $\sigma$ and $\sigma'$ composes
their denotations, with the seam $\mathsf{Coll}$ as the transfer function between
them: $\llbracket\sigma'\rrbracket\circ\llbracket\mathsf{Coll}\rrbracket\circ
\llbracket\sigma\rrbracket$. Cost is therefore a homomorphism from schedules to
$(\max,+)$ transfer functions, on the same composition the algebra already uses.
The step makespan is this composed function applied to the microbatch injection
times, and its steady-state period is the function's $(\max,+)$ eigenvalue, the
bottleneck stage. Transpose needs no special case: the warmup slack $w$ is just
when the LLM region becomes ready on each rank, part of the arrival times the
composition already carries, so the spill $\max(0,E-w)$ is what its transfer
function returns rather than a separate rule. For our acyclic, single-orientation
$\textit{order}$, the makespan read this way equals the operational longest path
above, the standard agreement between longest paths in a timed event graph and
$(\max,+)$ linear behavior~\cite{baccelli1992synchronization}. We state this as a
remark and leave the contention-aware proof, which threads the per-rank state
through the composition, to future work.

\paragraph{Scope.}
Mpipe handles the main, source-side form of MLLM data variance: the encoder
workload. The schedule algebra (\S\ref{sub:algebra}) moves it into pipeline
bubbles as a static schedule that does not depend on the modality mix. Three cases
fall outside this, and the same algebra states each one. \emph{(i)} When a
microbatch's encoder work is larger than the warmup slack, part of it spills onto
the critical path; the owner map reduces this spill but cannot remove it.
\emph{(ii)} Variance in the LLM \emph{sequence} length makes the backbone stages
themselves data-dependent. This is interior to the pipeline and is handled by
sequence bucketing or packing, not by placement. \emph{(iii)} Output-side
\emph{generators} are pipeline sinks. Their input comes from the backbone, not
from the start of the step, so they cannot be transposed into a bubble. In the
algebra they are just another $\mathsf{Sharded}$ region, for example
$\langle\mathsf{transpose},\mathsf{1f1b},\mathsf{1f1b}\rangle$, and are pipelined
rather than moved. Runtime systems cover these cases with per-iteration
search~\cite{xue2026dip} or data-dependent parallelism~\cite{megascaleomni}.
Mpipe keeps the schedule static and loss-preserving. It trades that coverage for
near-zero scheduling overhead and reproducible runs, and we evaluate it in the
encoder-bound case where this trade pays off.

\section{Experiments}\label{sec:exp}
We conduct experiments to evaluate the overall performance of Mpipe on MLLM
training workloads.

\textbf{Environment.} We implement and evaluate Mpipe in Hyper-Parallel. The
evaluation runs on a 512-device CloudMatrix384~\cite{zuo2025cloudmatrix384}.
Mpipe runs inside the framework's existing pipeline runtime: $\operatorname{derive}$
lowers a schedule to the placement, collectives, and order of \S\ref{sub:algebra},
and the $\mathsf{Replicated}\to\mathsf{Sharded}$ gather uses CANN/HCCL, with no
change to the LLM backbone's pipeline kernels.

\textbf{Experiment A: production-scale MLLM workload.} We design this
experiment to simulate the effect of Mpipe under practical industry MLLM
training workloads. The model is deployed on the full 512-device environment and
uses the ViT encoder from Qwen2-VL~\cite{wang2024qwen2vl} together with a
DeepSeek-V3~\cite{deepseekv3} LLM backbone. The dataset is an internal
multimodal dataset. Since
DistTrain has not released its implementation, we
construct a DistTrain-like baseline in our environment: one pipeline stage is
allocated to the encoder with a customized parallel strategy, while the LLM
backbone applies conventional 5D parallelism. Table~\ref{tab:experiment-a}
shows that Mpipe achieves 1.21$\times$ speedup over this baseline.

\begin{table}[t]
\centering
\caption{Experiment A: average end-to-end step time on a production-scale MLLM workload.}
\label{tab:experiment-a}
\resizebox{\columnwidth}{!}{%
\begin{tabular}{l r r r}
\toprule
Model & Baseline Avg. Step Time & Mpipe Avg. Step Time & Speedup \\
\midrule
ViT+DeepSeek & 16.26 s & 13.42 s & 1.21$\times$ \\
\bottomrule
\end{tabular}%
}
\end{table}

\textbf{Experiment B: generality study.} To quickly validate the
generality of Mpipe, we also run an 8-device experiment with a
recently released Qwen3.5-based MLLM~\cite{qwen2026qwen35omni} on the
CapsFusion image-text dataset~\cite{yu2023capsfusion}. The baseline follows
Megatron-LM~\cite{shoeybi2019megatron}. As shown in
Table~\ref{tab:experiment-b}, Mpipe reduces the average step time from 11.01 s
to 4.07 s, achieving 2.70$\times$ speedup.

\begin{table}[t]
\centering
\caption{Experiment B: average end-to-end step time on an 8-device workload.}
\label{tab:experiment-b}
\resizebox{\columnwidth}{!}{%
\begin{tabular}{l r r r}
\toprule
Model & Baseline Avg. Step Time & Mpipe Avg. Step Time & Speedup \\
\midrule
Qwen3.5 MLLM & 11.01 s & 4.07 s & 2.70$\times$ \\
\bottomrule
\end{tabular}%
}
\end{table}

\textbf{Analysis.} The relative speedup of transpose depends on the encoder's
share of the step and on how much warmup bubble is available to absorb it
(\S\ref{sub:cost}). This is why the two settings differ in magnitude: in
Experiment~A the large backbone dominates the step, so relocating the smaller
encoder yields $1.21\times$, whereas in the controlled setting of Experiment~B
the encoder is a larger share of the step and the gain reaches $2.70\times$.
Because the two settings also differ in model and baseline, they are not a
controlled scaling study; a sweep that isolates scale and the resulting
encoder-to-bubble ratio is left to future work. In both cases the gain comes from
a single static schedule with no per-iteration scheduling cost, and transpose
leaves the training loss unchanged, since it only relocates computation.

\section{Related Work}\label{sec:related}

\paragraph{Pipeline schedules.}
Pipeline parallelism partitions a model into stages and streams microbatches
through them. GPipe~\cite{huang2019gpipe} flushes each batch; 1F1B and
interleaved (VPP) schedules reduce the bubble at fixed memory; recent work pushes
further with wave-like (Hanayo~\cite{hanayo}), zero-bubble
(DualPipeV~\cite{dualpipev}), and bidirectional (DualPipe~\cite{deepseekv3})
schedules. Mpipe does not propose one more schedule in isolation. It treats these
as points in a single algebra: 1F1B, GPipe, and interleaved VPP are instances,
and the recent schedules are short, named extensions of it (\S\ref{sub:algebra}).
Within that algebra, $\langle\mathsf{transpose},\mathsf{1f1b}\rangle$ is the
schedule that handles the encoder/LLM heterogeneity of MLLM training.

\paragraph{MLLM training systems.}
DistTrain~\cite{zhang2025disttrain}, Optimus~\cite{feng2025optimus}, and
DIP~\cite{xue2026dip} target encoder/backbone heterogeneity through spatial
partitioning, temporal bubble profiling, and runtime load balancing
respectively; OrchMLLM~\cite{zheng2025orchmllm} and
MegaScale-Omni~\cite{megascaleomni} adapt resources and parallelism to dynamic
multimodal workloads in production. As discussed in \S\ref{sec:limits}, all of
these keep the encoder coupled inside the backbone pipeline and pay for
heterogeneity and dynamism at runtime. Mpipe instead relocates the encoder out of
the pipeline with a single static placement, which removes the per-iteration cost
and is invariant to the modality mix.

\paragraph{Schedules as composable skeletons.}
Mpipe's view of schedules as composable parallel patterns follows the line of
parallel algorithmic skeletons~\cite{cole1989skeletons}, which hide orchestration
behind structured collectives rather than raw point-to-point
messaging~\cite{gorlatch2004sendrecv}, and the BSP-structured, cost-predictable
skeleton libraries~\cite{loulergue2005bsml,javed2009osl,li2012sgl}. The
difference is one of level: those compose skeletons to express an algorithm,
whereas we compose them to express the execution schedule of a fixed model.

\section{Conclusion}
MLLM training exposes component heterogeneity and dynamic encoder workloads that
make conventional LLM parallelism less effective at scale. 
Mpipe casts parallel schedules as a small schedule algebra and derives their concrete pipeline behavior, including transpose, a heterogeneous parallel schedule for MLLM training. Experiments on Ascend 910C NPU clusters show that Mpipe reduces end-to-end
training step time on both small-scale and 512-card settings.

\paragraph{Limitations and future work.}
When per-microbatch encoder work exceeds the warmup slack, part of it spills onto
the critical path; the static owner map reduces this spill but cannot remove it,
and a loss-preserving metadata-guided reordering that closes the residual is left
to future work. We also instantiate transpose only on the two-region encoder/LLM
cut with 1F1B; the algebra supports richer schedules (a fold toward Hanayo, an
event split toward DualPipeV, a bidirectional feed toward DualPipe), whose
evaluation we leave to future work.

\bibliographystyle{icml2026}
\bibliography{references}

\end{document}